\documentclass[submission,copyright,creativecommons]{eptcs}
\providecommand{\event}{QPL 2025} 

\usepackage{iftex}

\ifpdf
  \usepackage{underscore}         
  \usepackage[T1]{fontenc}        
\else
  \usepackage{breakurl}           
\fi
 
\input{modalities.sty}

\title{Proto-Quipper with Reversing and Control}
\author{Peng Fu
\institute{University of South Carolina, USA}
\and
Kohei Kishida
\institute{University of Illinois Urbana-Champaign, USA}
\and 
Neil J. Ross \qquad\qquad Peter Selinger
\institute{Dalhousie University, Canada}
}
\def\titlerunning{Proto-Quipper with Reversing and Control}
\def\authorrunning{P. Fu, K. Kishida, N.J. Ross \& P. Selinger}
\begin{document}
\maketitle

\begin{abstract}
The quantum programming language Quipper supports circuit operations
such as reversing and controlling certain quantum
circuits. Additionally, Quipper provides a function called
\textit{with-computed}, which can be used to program circuits of the
form $g; f; g^{\dagger}$. The latter is a common pattern in quantum
circuit design.  One benefit of using with-computed, as opposed to
constructing the circuit $g ; f; g^{\dagger}$ directly from $g$, $f$,
and $g^{\dagger}$, is that it facilitates an important
optimization. Namely, if the resulting circuit is later controlled,
only $f$ needs to be controlled; the circuits $g$ and $g^{\dagger}$
need not even be controllable.

In this paper, we formalize a semantics for reversible and
controllable circuits, using a dagger symmetric monoidal category $\R$
to interpret reversible circuits, and a new notion we call a
\textit{controllable category} $\N$, which encompasses the control and
with-computed operations in Quipper. We extend the language
Proto-Quipper with reversing, control and the with-computed operation.
Since not all circuits are reversible and/or controllable, we use a
type system with modalities to track reversibility and
controllability. This generalizes the modality of
Fu-Kishida-Ross-Selinger 2023.  We give an abstract categorical
semantics, and show that the type system and operational semantics are
sound with respect to this semantics.
\end{abstract}

\section{Introduction}

The goal of this paper is to devise a strongly typed, functional
programming language that can formally treat the quantum operations of
reversing and control.  Many quantum algorithms take essential
advantage of reversing, control, and their associated operation
``with-computed''.  At the same time, some operations, such as state
preparations and measurements, can make circuits irreversible or
uncontrollable.  It is therefore desirable to equip quantum
programming languages with a type system that can prevent the user
from trying to reverse the irreversible or to control the
uncontrollable. In this paper, we incorporate
reversing, control, and with-computed into the language
\emph{Proto-Quipper}.

Proto-Quipper is a family of programming languages that aims to
provide a formal syntax and semantics for various features of the
embedded quantum programming language Quipper
\cite{GLRSV2013-rc,GLRSV2013-pldi}.  Like Quipper, Proto-Quipper is a
quantum circuit description language. When a Proto-Quipper program is
run, it outputs a quantum circuit; this generated circuit can later be
executed on a quantum computer. We refer to these two runtimes as
``circuit generation time'' and ``circuit execution time'',
respectively.  There are various errors that programmers can introduce
by attempting the impossible---%
e.g., to duplicate linear data, to apply circuit operations to a
program that is not a circuit, etc.---%
but since Proto-Quipper is strongly typed, it can catch such errors at
compile time, rather than at circuit generation time or circuit
execution time.  Previous research on Proto-Quipper includes providing
a syntax and semantics for circuit boxing
\cite{RS2017-pqmodel,ross2015algebraic}, supporting recursion
\cite{BMZ2018}, combining linear types and dependent types
\cite{FKRS2020-dpq-tutorial,FKS2020-lindep}, and incorporating dynamic
lifting \cite{ColledanL22,FKRS-2022,FKRS-2023,LeePVX21}.  This paper
defines a new variant Proto-Quipper-C, which extends the type system
of Proto-Quipper with reversibility and controllability.

\subsection{Reversing, control, and with-computed}

In Quipper, the \texttt{controlled} function inputs a circuit and
returns the controlled version of the circuit. The \texttt{reverse}
function inputs a circuit and returns its adjoint. The program
\texttt{(with\_computed$\,\,g\,\,f$)} first performs a circuit $g$, then a
circuit $f$, and finally the adjoint of $g$.
\[
\Qcircuit @C=0.5em @R=.2em @!R {
  & \multigate{2}{g}& \qw & \multigate{2}{f} & \qw &
  \multigate{2}{g^{\dagger}} & \qw\\
  & \ghost{g} & \qw & \ghost{f} & \qw & \ghost{g^{\dagger}} & \qw\\
  & \ghost{g} & \qw & \ghost{f} & \qw & \ghost{g^{\dagger}} & \qw \\
}
\]

The \textit{with-computed} circuit pattern is very common in quantum computing. For example,
Bennett's method for constructing reversible classical circuits \cite{bennett1973logical}
uses this pattern to initialize ancillary qubits and uncompute
them; and the quantum Fourier transform (QFT) addition \cite{draper2000addition} conjugates a series of controlled rotation gates by the QFT. This pattern is used so often that
Quipper implements the following optimization: when controlling a quantum circuit $g; f; g^{\dagger}$ generated by the with-computed function, it suffices to control the circuit $f$.

\[
  \begin{array}{lll}
    \Qcircuit @C=0.5em @R=.2em @!R {
         &\qw & \qw & \ctrl{1} & \qw & \qw & \qw \\
      & \multigate{2}{g}& \qw & \multigate{2}{f} & \qw & \multigate{2}{g^{\dagger}} & \qw\\
      & \ghost{g}       & \qw & \ghost{f}        & \qw & \ghost{g^{\dagger}}        & \qw\\
      & \ghost{g}       & \qw & \ghost{f}        & \qw & \ghost{g^{\dagger}}        & \qw \\
    }
          & 
        \raisebox{-2.6em}{=} &
    \Qcircuit @C=0.5em @R=.2em @!R {
       &\ctrl{1} & \qw & \ctrl{1} & \qw & \ctrl{1} & \qw \\
      & \multigate{2}{g}& \qw & \multigate{2}{f} & \qw & \multigate{2}{g^{\dagger}} & \qw\\
      & \ghost{g} & \qw & \ghost{f} & \qw & \ghost{g^{\dagger}} & \qw\\
      & \ghost{g} & \qw & \ghost{f} & \qw & \ghost{g^{\dagger}} & \qw \\
    }                              
  \end{array}
\]

\noindent In fact, the circuits $g$ and $g^{\dagger}$ need not even be
controllable.  Thus the left hand side of the above circuit identity
is more general than the right hand side.  It is also more efficient,
since it uses a smaller number of controlled gates. Note that in
quantum computing, a controlled gate may require substantially more
resources than the non-controlled version.  For example, a common way
of measuring resources is the \textit{T-count}, i.e., the number of
$T$-gates required to implement a gate. While a $T$-gate itself has
$T$-count $1$, a controlled $T$-gate has a $T$-count of at least $5$.

\subsection{Modalities for reversing and control}

In Quipper, controlling an uncontrollable circuit or reversing a
non-reversible circuit will give rise to runtime errors. We want to
incorporate reversing and control in Proto-Quipper in such a way that
erroneously controlling or reversing a circuit is detected as a typing
error at compile time. We achieve this by introducing a notion of
modality $\alpha \in \{0, 1, 2\}$. Here, the modality $2$ is
associated with circuits that are both controllable and reversible,
for example a Hadamard gate. The modality $1$ is associated with
circuits that are reversible but not controllable, for example an
initialization gate. (Recall that ``reversing'' means taking the
adjoint of a circuit, not necessarily the inverse. For a more detailed discussion of this circuit model, see \cite[Sec.~4.2]{GLRSV2013-pldi}.)
The modality $0$
is associated with circuits that are neither controllable nor
reversible, for example a measurement gate.
Note that because all controllable quantum
gates are unitary and therefore reversible, 
we do not include a modality for circuits that are controllable but not reversible.

We can use modalities to annotate the type of a circuit. E.g.,
the values of the type ${\Circ_{\alpha}(S, U)}$ are circuits with input
type $S$, output type $U$ and a modality annotation $\alpha$.  One
important feature of modalities is that they are composable.  Suppose
we want to compose a circuit $\Circ_{\alpha}(S, U)$ with a circuit
$\Circ_{\beta}(U, S')$. The resulting circuit will have type
$\Circ_{\alpha \wedge \beta}(S, S')$, where
$\alpha \wedge \beta = \mathrm{min}(\alpha, \beta)$.  This means that
as long as we know the modalities for the basic gate sets, we can
devise a type system to track the modalities of the circuits
constructed from basic gates.  Moreover, the reversing, control, and
with-computed operations can be given the following types:

\[
\begin{array}{r@{~~}l@{~~}ll}
  \controlled_{S}   & : & \Circ_{\textcolor{blue}{2}}(U, U) \to \Circ_{\textcolor{blue}{2}}(S\otimes U, S\otimes U). & 
  \\
  \reverse  & : & \Circ_{\textcolor{blue}{\alpha}}(U, S) \to \Circ_{\textcolor{blue}{\alpha}}(S, U), &  
                                                                                                    \text{where $\alpha > 0$}.     
\\
  \withComputed & : & \Circ_{\textcolor{blue}{\alpha}}(U, S) \times \Circ_{\textcolor{blue}{2}}(S, S) \to \Circ_{\textcolor{blue}{2}}(U, U), & 
  \text{where $\alpha > 0$}. 
\end{array}
\]
\noindent Note that $\withComputed$ requires its first argument to be
at least reversible, and the resulting circuit is again controllable.  

Our modal type system makes it possible to detect errors, like
controlling an uncontrollable circuit, at compile time. In a practical
implementation (such as the one available from {\cite{prototype}}), it
is moreover possible for the type checker to automatically infer
modalities, so that the programmer does not need to explicitly specify
them in the source code.  See Appendix~\ref{app:controlling-ccz} for
some sample code from our prototype implementation.

\subsection{Related work}

The reverse, control and with-computed operations originally appeared
in Quipper \cite{GLRSV2013-pldi}, but were lacking a formal semantics.
In our work on Proto-Quipper \cite{FKRS-2022,FKRS-2023}, we gave type
systems with modalities for a feature called dynamic lifting, but not
for reversing and control.

Many practical quantum programming languages, such as Qiskit and Cirq
also include the concepts of reversibility and controllability.
However, since these languages are imperative, errors of reversibility
or controllability are treated as runtime exceptions
{\cite{cirq,qiskit}}. The functional quantum programming language
QWire \cite{paykin2017qwire} supports circuit reversal, but not the
control and with-computed operations. The language Silq \cite{silq}
has type annotations \textbf{qfree} and \textbf{mfree}, which denote
classical functions and functions that do not use measurement,
respectively. Silq's core language supports reversing, and the control
operation is supported by an if-then-else construct.  The main
difference between Silq and Proto-Quipper is that Silq is not a
circuit description language and does not have a notion of boxed
circuits.

We are only aware of one language besides Quipper that has a
with-computed operation, namely ProjectQ
\cite{haner2018software,steiger2018projectq}.  However, ProjectQ is
not a strongly-typed language.

\subsection{Contributions}

In this paper, we formalize a notion of controllable category. We
extend the Proto-Quipper type system of \cite{FKRS-2023} with
modalities for reversing and control and define an operational
semantics. We provide safety properties that guarantee that a
well-typed program is free of run-time errors. We also give an
axiomatization of an abstract categorical model of Proto-Quipper with
reversing, control, and the with-computed operation, and show that the
operational semantics is sound for it.

\section{Controlling a permutation circuit}
\label{ssec:cont-perm}

In Proto-Quipper, there are two ways to represent the swap
operation. One way is by permuting the logical order of the 
circuit outputs without using any gates, as in the following diagram. 
\begin{equation}\label{eqn:perm-1}
  \vcenter{
    \Qcircuit @C=2.1em @R=1.5em {
      \lstick{\text{input 1} = x} & \qw & \rstick{x = \text{output 2}}\qw \\
      \lstick{\text{input 2} = y} & \qw & \rstick{y = \text{output 1}}\qw \\
    }
  }
\end{equation}

\noindent The above circuit is generated by the following Proto-Quipper program.
{\small
\begin{verbatim}
f : !(Qubit * Qubit -> Qubit * Qubit)
f input = let (x, y) = input in (y, x)
\end{verbatim}
}
The other way is by using an explicit swap gate, as in the following diagram.
\begin{equation}\label{eqn:perm-2}
  \vcenter{
    \Qcircuit @C=2.1em @R=1.5em {
      \lstick{\text{input 1} = x} & \qswap & \rstick{y = \text{output 1}}\qw \\
      \lstick{\text{input 2} = y} & \qswap \qwx & \rstick{x = \text{output 2}}\qw 
    }
  }
\end{equation}

\noindent The corresponding Proto-Quipper program is the following.
{\small
\begin{verbatim}
g : !(Qubit * Qubit -> Qubit * Qubit)
g input = let (x, y) = input in Swap x y
\end{verbatim}
}
\noindent These circuits are semantically equivalent: each of them
sends $(x,y)$ to $(y,x)$.  However, care must be taken when
controlling these circuits, since controlling them naively may produce
the following non-equivalent circuits.

\begin{equation}\label{eqn:perm-3}
  \small
    \hspace{0.75in}
    \vcenter{
      \Qcircuit @C=2.1em @R=1.5em {
        \lstick{\text{(a)}\quad\text{control} = c} & \qw & \rstick{c}\qw \\
        \lstick{\text{input 1} = x} & \qw & \rstick{x = \text{output 2}}\qw \\
        \lstick{\text{input 2} = y} & \qw & \rstick{y = \text{output 1}}\qw \\
      }
    }
    \hspace{1.3in}
    \hspace{0.75in}
    \vcenter{
      \Qcircuit @C=2.1em @R=1.5em {
        \lstick{\text{(b)}\quad\text{control} = c}& \ctrl{1} & \rstick{c}\qw\\
        \lstick{\text{input 1} = x} & \qswap & \rstick{y = \text{output 1}}\qw \\
        \lstick{\text{input 2} = y} & \qswap \qwx & \rstick{x = \text{output 2}}\qw \\
      }
    }
    \hspace{0.75in}
\end{equation}
\smallskip
      
While the second circuit is correct, the first one is not. Indeed, the
first circuit will send input 1 to output 2 independently of the state
of the control qubit. This is not the correct behavior, since the
swapping should only take place when the control qubit is in the state
$\ket{1}$. When the control qubit is in the state $\ket{0}$, the
functions $f$ and $g$ should both behave like the identity function on
\verb!Qubit * Qubit!.

Therefore, the programming language must be aware not only of the
circuit, but also of the implicit permutation of any qubits performed
in the language. If such an operation is controlled, explicit swap
gates must sometimes be inserted. Proto-Quipper-C does this correctly,
whereas the original Quipper implementation did not. In
Proto-Quipper-C, controlling circuit {\eqref{eqn:perm-1}} and
circuit {\eqref{eqn:perm-2}} both result in {\eqref{eqn:perm-3}}(b),
not {\eqref{eqn:perm-3}}(a).

\section{A formal model of reversing, control and with-computed}
\label{sec:formal-model}

Before we can define a programming language for building quantum
circuits, we must specify what a quantum circuit is. However, there
exist many different classes of circuits; for example, they differ by
what data they can manipulate (only qubits, or also classical bits,
and/or more exotic objects like qutrits), what the built-in gates are
(for example, the Clifford+$T$ gate set, Toffoli+Hadamard, rotations
by arbitrary angles), whether or not qubit initialization and
termination is supported, whether measurement is considered as a gate,
and so on. Therefore, rather than tailoring our programming language
to a specific class of circuits, we make both the language and its
operational and denotational semantics \emph{parametric} on a given
class of circuits. This is analogous to making a classical programming
language parametric on some signature of built-in operations. For us,
quantum circuits are abstractly given as the morphisms of monoidal
categories with certain properties, which we now specify.

We start from a given symmetric monoidal category $\m$. In practice,
the objects of $\m$ are typically generated from wire types such as
$\Bit$ and $\Qubit$, and the morphisms are quantum circuits generated
from a finite set of gates. In the following, we define what we mean
by a reversible subcategory of $\m$.  Recall that a dagger symmetric
monoidal category is a symmetric monoidal category equipped with a
contravariant, identity-on-objects, involutive functor such that for
all morphism $f:A\to B$, we have $f^{\dagger}:B\to A$. Moreover, the
dagger functor is required to be compatible with the symmetric
monoidal structure \cite{selinger2007dagger}.

\begin{definition}
  By a \emph{reversible subcategory} of $\m$, we mean a dagger
  symmetric monoidal category $\R$, together with an
  identity-on-objects, faithful, symmetric (strong) monoidal functor
  $I:\R\to\m$. We usually regard $I$ as an inclusion functor, i.e., we
  regard $\R$ as a subcategory of $\m$.
\end{definition}

A morphism $f$ of $\R$ is called \emph{unitary} if
$f \circ f^{\dagger} = \id$ and $f^{\dagger} \circ f = \id$. Note that
we do not require all morphisms of $\R$ to be unitary. Thus, our
notion of reversible morphism means a morphism that has an adjoint,
not necessarily an inverse.  The dagger functor provides a semantics
for the operation of reversing a quantum circuit.

In order to capture the notion of control and with-computed, we define
a notion of controllable category $\N$. To motivate the following
definition, we should first clarify some use cases. The most common
example of control from quantum computing is the control usually
represented by a black dot ``$\bullet$''. This kind of control
``fires'' the gate if the control qubit is in state $\ket{1}$, and
acts as the identity otherwise. There are also other ways of
controlling a qubit; for example, the white dot ``$\circ$'' fires if
the control qubit is in state $\ket{0}$. A more general kind of control
in quantum computing is as follows: Let $V, W$ be Hilbert spaces,
let $K$ be a subspace of $V$, and let $G:W\to W$ be a gate. We define
the $K$-controlled gate $\ctrled_{K}(G) : V\otimes W\to V\otimes W$
by $\ctrled_{K}(G)(v\otimes w) = v\otimes G(w)$ if $v\in K$ and
$\ctrled_{K}(G)(v\otimes w) = v\otimes w$ if $v\in K^{\bot}$,
extended to all other states by linearity. Note that this general
notion of control not only encompasses the black and white dot, but
also many other kinds; for example, it is possible to control a gate
on the state $\ket{+}$, to control on multiple qubits and/or classical
bits, in which case the subspace $K$ may or may not be
$1$-dimensional. Even the trivial control is possible, namely when
$K=V$; in this case the gate always ``fires''.

Working in an abstract monoidal category, we will not talk about
subspaces and their orthogonal complements. Instead, we consider a monoid $\Cc$ of \textit{controls}, regarded as
a discrete monoidal category, together with a monoidal functor
$F:\Cc\to\N$. We write $\uu{K}$ for $F(K)$. The idea is that each $K\in \Cc$ specifies a kind of control on the object $\uu{K}\in\N$.

\begin{definition}
  Let $\R$ be a dagger symmetric monoidal category. By a
  \emph{controllable category} of $\R$, we mean a dagger symmetric
  monoidal category $\N$ with the same objects as $\R$, and equipped with
  the following structure:

  \begin{enumerate}
  \item\label{cond:b} For all $A, B \in \N$, a function
    $- \action- : \R(B, A) \times \N(A, A) \to \N(B, B)$ (also denoted by $\withComputed$) such that:
    \[\id \action h = h\]
    \[ (g_{1} ; g_{2}) \action h =   g_{1} \action (g_{2} \action h).\]
    \[(g_{1}\otimes g_{2}) \action (h_{1}\otimes h_{2}) = (g_{1}\action h_{1}) \otimes (g_{2}\action h_{2})\]
    \[(g \action h)^{\dagger} = g \action h^{\dagger}\]

  \item\label{cond:C} A discrete monoidal category $\Cc$ of {\em
      controls}, together with a monoidal functor mapping $K\in\Cc$ to
    $\uu{K}\in\N$.
    
  \item\label{cond:d} For all $A \in \N$ and $K\in \Cc$, a function
    $\ctrled_{K} : \N(A, A) \to \N(\uu{K}\otimes A, \uu{K}\otimes A)$ such
    that the following hold:
    \[\ctrled_{K}(\id_{A}) = \id_{\uu{K}\otimes A}\]
    \[
      \begin{tikzcd}
        \uu{I}\otimes A \arrow[d, "\iso"] \arrow[rr, "{\ctrled_{I}(h)}"] && \uu{I}\otimes A \arrow[d, "\iso"]\\
        A \arrow[rr, "h"]& & A,
      \end{tikzcd}
    \]
    \[
      \begin{tikzcd}
        (\uu{A \otimes B})\otimes C
        \arrow[d, "\iso"]
        \arrow[rrrr, "{\ctrled_{A\otimes B}(h)}"]
        &&&& (\uu{A \otimes B})\otimes C
        \arrow[d, "\iso"]\\
        \uu{A} \otimes (\uu{B} \otimes C)
        \arrow[rrrr, "{\ctrled_{A}(\ctrled_{B}(h))}"]
        &&&& \uu{A} \otimes (\uu{B} \otimes C),
      \end{tikzcd}
    \]
    \[
      \begin{tikzcd}
        (\uu{A \otimes B})\otimes C
        \arrow[d, "\iso"]
        \arrow[rrrr, "{\ctrled_{A\otimes B}(h)}"]
        &&&& (\uu{A \otimes B})\otimes C
        \arrow[d, "\iso"]\\
        (\uu{B \otimes A})\otimes C
        \arrow[rrrr, "{\ctrled_{B\otimes A}(h)}"]
        &&&& (\uu{B \otimes A})\otimes C.
      \end{tikzcd}
    \]
  \item\label{cond:c} An identity-on-objects, strict dagger symmetric monoidal functor $G : \N \to \R$ such that
    \[ G(g \action h) = g; G(h) ; g^{\dagger}
    \]
    
  \end{enumerate}
\end{definition}

\noindent \textbf{Remarks}.
\begin{itemize}

\item The functor $G$ gives an intended interpretation for the
  with-computed function. It captures the common quantum computation pattern $g; h; g^{\dagger}$.
  
\item The functor $G$ preserves the equalities of condition (\ref{cond:b}). For example, $G(\id\action h) = \id; G(h); \id^{\dagger} = G(h)$. 
  Note that we do not require $g \action \id_{A} = \id_{B}$, because this would imply $g; g^{\dagger} = g; \id_{A}; g^{\dagger} = G(g \action \id_{A}) = G(\id_{B}) = \id_{B}$, which is not an assumption we make for $\R$. For similar reasons,
  we do not require $g \action (h_{1} ;  h_{2}) = (g \action h_{1}) ; (g \action h_{2})$. 

\end{itemize}

To illustrate how the above notions are applicable to quantum
computing, we can consider a ``standard model'' of these axioms,
parameterized by a class of quantum circuits. Suppose the circuits are
generated by some types of wires (for example, $\Qubit$ and $\Bit$)
and some set of gates. Suppose that some distinguished subset of the
gates are reversible, and a further subset of these are
controllable. Also suppose that the class of circuits is equipped with
some primitive control gadgets, for example, a black-dot and a
white-dot control on qubits, and that for each gate, all of its
controlled versions are also gates. The category $\m$ consists of all
circuits, arranged into a symmetric monoidal category (the symmetric
monoidal structure permits the crossing of wires, but such crossings
are not considered to be gates). Let $\R$ be the subcategory of
circuits that are made from reversible gates. Finally, the structure
of $\N$ is somewhat interesting. The morphisms of $\N$ are circuits
made from gates of two colors, say red and green. Each red gate is
reversible and each green gate is controllable (and in particular,
each controllable gate comes in two different colors). All morphisms
in $\N$ are reversible, and reversing does not change color. The
objects of the monoidal category $\Cc$ are sequences of control
gadgets, for example $\bullet\circ\circ$, and
$\uu{\bullet\circ\circ} = \Qubit\otimes\Qubit\otimes\Qubit$. The
control function adds controls to the green gates.  The with-computed
function takes any morphism $g$ of $\R$ and any morphism $f$ of $\N$,
and produces $g; f; g^{\dagger}$, where $g$ and $g^{\dagger}$ have
been colored red. The functor $G$ is the forgetful functor that
forgets the colors.

\section{A type system for reversing, control and with-computed}
\label{type-system}

In this section, we define the syntax and type system of
Proto-Quipper-C, a version of Proto-Quipper with support for
reversing, control and with-computed. We assume that we are given
three fixed small categories $\m$, $\R$, and $\N$ as specified in the
previous section. The type system and (in later sections) the
operational semantics and categorical semantics are all parameterized
by this choice.

\subsection{Syntax}

The syntax of Proto-Quipper-C is shown in Figure~\ref{fig:syntax}.
It is similar to the one specified in \cite{FKRS-2023}, except
for the highlighted terms and types.

\begin{figure}
  \[\small
    \def\arraystretch{1.1}
    \begin{array}{llll}
      \textit{Modalities} & \alpha, \beta  &::=&  \greybox{2}\mid \greybox{1} \mid \greybox{0}
      \\ 
      \textit{Types} & A, B  &::=&  \Unit \mid  W \mid  \Bool\mid  {!}_\alpha A 
                                   \mid A \multimap_\alpha B
        \\
        &&& {~}\mid  \greybox{\Circ_{\alpha}(S, U)} \mid  A \otimes B
      \\ 

      \textit{Parameter Types} & P, R  &::=& \Unit \mid  \Bool
      \mid {!}_{\alpha} A \mid  \Circ_{\alpha}(S, U) \mid P \otimes R

      \\ 
      
      \textit{Simple Types} & S, U  &::=& \Unit \mid  W
      \mid S \otimes U
      
      \\ 
      
      \textit{Terms} & M, N  &::=& c\mid  x \mid \ell \mid \unit \mid \lambda x . M \mid  M N  \mid  (M, N) \mid  \mathsf{let}\ (x, y) = N\  \mathsf{in}\ M \\
                          &&& {~}\mid \mathsf{lift}\ M\mid  \force \, M\\
                          &&& {~}\mid \greybox{\circc(S,C,U)}\mid \mathsf{apply}(M, N) \mid  \mathsf{box}_U M \\
                          &&& {~}  \mid \greybox{\reverse M} \mid \greybox{\controlled_{K} M}\mid \greybox{\withCompute{M}{N}}\\
      
      \textit{Simple Terms} & a, b  &::=& \ell  \mid  \unit\mid  (a, b)

      \\ 
      \textit{Contexts} & \Gamma  &::=& \cdot \mid  x : A, \Gamma \mid  \ell : W, \Gamma

      \\ 
      \textit{Parameter contexts} & \Phi  &::=& \cdot \mid  x : P, \Phi.
      
      \\ 
      \textit{Label Contexts} & \Sigma  &::=& \cdot \mid  \ell : W, \Sigma 

      \\ 

      \textit{Values}
      & V &::=& c \mid x \mid  \ell \mid  \unit \mid  \lambda x. M  \mid  (V, V')\mid  \mathsf{lift}\ M 
      \mid \circc(S,C,U)

      \\ 
      
        \textit{Circuits}
        & \greybox{C, D} & \in & \greybox{\N(\interp{S}, \interp{U})} \mid \greybox{\R(\interp{S}, \interp{U})} \mid \greybox{\m(\interp{S}, \interp{U})}               
\end{array}
\]

\caption{The syntax of Proto-Quipper-C}
\label{fig:syntax}
    
\end{figure}

\textbf{Simple types.} The letter $W$ ranges over a set of
\textit{wire types}, which in typical applications will be types that
represent a single wire in a circuit, such as $\Qubit$ and
$\Bit$. Each wire type represents a particular object of the category
$\m$, and tensors of wire types are called {\em simple types}. We
write $\interp{S}$ to denote the object corresponding to $S$ in the
category $\m$, $\N$, or $\R$ (note that they all have the same
objects).

\textbf{Labels.} We also assume that $\ell$ ranges over a set of
\textit{labels}, which are distinct from variables and are used as
pointers to places in a circuit where a gate can be attached. Unlike
variables, labels cannot be substituted, and they can only have wire
types. A \textit{simple term} is essentially a tuple of labels. We call a typing
context that contains only labels a \textit{label context} (denoted by
$\Sigma$). If $\Sigma = \ell_1:W_1,\ldots,\ell_n:W_n$, we write
$\interp{\Sigma}$ for the object corresponding to $\Sigma$ in $\m$,
$\N$, or $\R$, i.e., $\interp{W_1}\otimes\ldots\otimes\interp{W_n}$.

\textbf{Parameter types.} A certain subset of the types are called
\textit{parameter types}. These are the types whose values can be
freely duplicated and discarded, whereas values of simple types are
resources. Also, the values of parameter types are computed at circuit
generation time, whereas the values of simple types are merely tuples
of labels, which are placeholders for values that will be computed at
circuit execution time. A typing context that contains only parameter
types is a \textit{parameter context} (denoted by $\Phi$).  We have
included $\Bool$ as a typical example of a parameter type, but in a
real language, one would of course have many more such types (such as
$\Nat$, $\Int$, sum types, etc.) For space reasons, we omit a comprehensive
treatment of sum types, including booleans.

\textbf{Terms.} As usual, $x$ ranges over a set of \textit{variables}.
The symbol $c$ is used for constant symbols of the language, such as
$\mathsf{True}$ and $\mathsf{False}$ for the type $\Bool$, as well as
other appropriate constants. The terms of the lambda calculus are
fairly standard. Lift and force are from linear lambda
calculus, and are used to construct and deconstruct a term of type
$!_{\alpha}A$. We will discuss the rest of the terms later.

\textbf{Modalities.} The symbols $\alpha$ and $\beta$ range over
modalities, which we take to be numbers in the lattice $2 > 1 > 0$.
Here, $\alpha = 2$ indicates that a circuit is reversible and
controllable; $\alpha=1$ indicates that a circuit is reversible but
not necessarily controllable, and $\alpha=0$ indicates a general
circuit that may be neither reversible nor controllable. The modality
in the type $!_{\alpha}A$ means that when its value is forced, it may
append a gate that has modality $\alpha$ to the then-current circuit.
Similarly, the modality in the type $A \multimap_{\alpha} B$ indicates
that when its value is applied to an argument, it may append a gate
that has modality $\alpha$.

\textbf{Circuits.} Since Proto-Quipper-C is a circuit description
language, we must have some terms in the language that represent
circuits. We write such terms as $\circc(S,C,U)$, where
$C:\interp{S}\to\interp{U}$ is a morphism in the appropriate circuit
category $\m$, $\R$, or $\N$. Note that in this context, category
theory is used not as a denotational semantics (that will be done in
Section~\ref{sec:denotational-semantics}), but as part of the
\textit{syntax} of the language. This is effectively the same as
adding a constant symbol for every possible circuit. We write
$\Circ_{\alpha}(S, U)$ for the type of quantum circuits with inputs
$S$ and outputs $U$ (which are simple types), subject to the modality
$\alpha$ as described above. The terms $\apply$ and $\boxt$ are used
for constructing and deconstructing circuits; they will be explained in
more detail in the section on the operational semantics. We include
the terms $\reverse M$, $\controlled_{K} M$, and $\withCompute{M}{N}$
for reversing, control and the with-computed operation. The subscript
$K$ ranges over the objects of the category $\Cc$ which is part of the
given structure of the control category $\N$. We moreover assume that
$\uu{K}$ is a simple type (rather than just an arbitrary object).

In a practical implementation, programmers will not usually write
values of the form $\circc(S,C,U)$ directly. Instead, basic quantum gates
such as the Hadamard gate will usually be predefined in a library, and
users will write programs to combine these into larger circuits.

\subsection{Typing rules}

\begin{figure}
    \[ \small
      \begin{tabular}{cc}
        \multicolumn{2}{c}{
        \infer[\textit{var}]
          {\Phi, x : A \vdash_2 x : A}{}
          \hspace{0.5in}
          \infer[\textit{label}]
            {\Phi, \ell : W \vdash_2 \ell : W}{}
          \hspace{0.5in}
          \infer[\textit{unit}]
            {\Phi\vdash_2 \unit : \Unit}{}
        }            
        \\
        \\

                \infer[\textit{lambda}]{\Gamma \vdash_{2} \lambda x . M : A \multimap_\alpha B}
        {\Gamma, x : A \vdash_{\alpha} M : B} %
                                                           &
      \infer[\textit{app}]
      {\Phi, \Gamma_1, \Gamma_2 \vdash_{\alpha_1 \wedge \alpha_2 \wedge \beta} M N : B}
      {\Phi, \Gamma_1 \vdash_{\alpha_1} M :  A \multimap_\beta B & \Phi, \Gamma_2 \vdash_{\alpha_2} N : A}
      \\
      \\
      \infer[\textit{lift}]{\Phi \vdash_2 \lift M : {!}_\alpha A}
      {\Phi \vdash_\alpha M : A}

      &
      \infer[\textit{pair}]
      {\Phi, \Gamma_1, \Gamma_2 \vdash_{\alpha_1 \wedge \alpha_2} (M, N) : A\otimes B}
      {\Phi,\Gamma_1 \vdash_{\alpha_1} M :  A & \Phi,\Gamma_2 \vdash_{\alpha_2} N : B}
      \\
      \\

\infer[\textit{force}]{\Gamma \vdash_{\alpha \wedge \beta} \mathsf{force}\ M : A}
        {\Gamma \vdash_\beta M :\ !_\alpha A}
      &
        \infer[\textit{let}]
        {\Phi, \Gamma_1, \Gamma_2 \vdash_{\alpha_1 \wedge \alpha_2} \mathsf{let}\ (x, y) = N\  \mathsf{in}\ M : C}
        {\Phi, \Gamma_1, x : A, y : B \vdash_{\alpha_1} M :  C &  \Phi, \Gamma_2 \vdash_{\alpha_2} N : A\otimes B}
      \\
      \\
      \infer[\textit{circ}]{\Phi \vdash_{2} \circc(S,C,U) :  \Circ_{\alpha}(S, U)}
      {
      \begin{array}{ll}
        C \in \D^{\alpha}(\interp{S} , \interp{U}) 
      \end{array}
      }
      &
        \infer[\textit{apply}]{\Phi, \Gamma_1, \Gamma_2 \vdash_{\alpha \wedge \beta \wedge \gamma} \mathsf{apply}(M, N) :  U}{\Phi, \Gamma_1 \vdash_\alpha M : \Circ_{\gamma}(S,U) & \Phi, \Gamma_2 \vdash_\beta N : S}
      \\
      \\

      \infer[\textit{box}]{\Gamma \vdash_\alpha \mathsf{box}_S \  M : \Circ_{\beta \wedge \gamma}(S,U)}{\Gamma \vdash_\alpha M : {!}_{\beta} (S \multimap_{\gamma} U)}&
         \infer[\textit{ctrl}]{\Gamma \vdash_\alpha \controlled_{K} M : \Circ_{2}(\uu{K}\otimes S,\uu{K}\otimes S)}{\Gamma \vdash_\alpha M : \Circ_{2}(S, S)}        
      \\
      \\
      \infer[\textit{rev}]{\Gamma \vdash_\alpha \reverse M : \Circ_{\gamma}(U, S)}{\Gamma \vdash_\alpha M : \Circ_{\gamma}(S, U) & \gamma > 0} 
      &
        \infer[\textit{wc}]{\Phi, \Gamma_1, \Gamma_2 \vdash_{\alpha \wedge \beta} \withCompute{M}{N} : \Circ_2(U, U)}
        {
        \begin{array}{ll}
          \Phi, \Gamma_1 \vdash_\alpha M : \Circ_{\gamma}(U, S) &  \gamma > 0\\
          \Phi, \Gamma_2 \vdash_\beta N : \Circ_2(S, S)
        \end{array}
      }   
     \end{tabular}
   \]
   \caption{The typing rules}
   \label{fig:typing}
\end{figure}

The typing rules are shown in Figure~\ref{fig:typing}. Typing
judgments are of the form $\Gamma \vdash_{\alpha} M : A$. Here, the
modality $\alpha$ asserts that during the evaluation of the term $M$,
gates of modality $\alpha$ may be appended to the current circuit.
Recall that we write $\alpha \wedge \beta$ to denote the greatest
lower bound, or $\mathrm{min}(\alpha, \beta)$.
To facilitate the statement of the typing rules, we treat
the contexts as unordered lists.

Except for the modalities, the typing rules are similar to the ones in
{\cite{FKRS-2023}}. We make the following observations about the
modalities. Since a value cannot append any gates, it does not change
the current circuit state. Therefore values always have modality $2$,
i.e., $\Gamma \vdash_{2} V : A$. The \textit{lambda} and \textit{lift}
rules store the modality $\alpha$ of $M$ in the respective types
$A \multimap_{\alpha} B$ and $!_{\alpha}A$.  Similarly, in the typing
rule \textit{box}, the modalities $\beta, \gamma$ jointly determine
the modality of the circuit type. Conversely, in the rules
\textit{app}, \textit{force}, and \textit{apply}, the modality in the
types $A \multimap_{\alpha} B$, $!_{\alpha}A$, and
$\Circ_{\gamma}(S,U)$ affects the modality of the current term. In the
rule \textit{circ}, we write $\D^{\alpha}$ to mean $\m$ if
$\alpha = 0$, and $\R$ if $\alpha = 1$, and $\N$ if $\alpha = 2$. The
\textit{rev} rule requires the circuit to be reversible, i.e.,
$\gamma$ must be at least 1. Similarly, the rule \textit{ctrl}
requires a controllable circuit, i.e., modality 2.  It also requires
the circuit to have matching input and output types, i.e.,
$\Circ_2(S, S)$ rather than $\Circ_2(S, U)$. Finally, the typing rule
for with-computed requires the term $N$ to be a controllable circuit,
whereas the term $M$ only needs to be reversible. The resulting term
$\withCompute{M}{N}$ is again a controllable circuit.

\section{Operational semantics}
\label{op-sem}
In this section, we give an operational semantics of our language. We follow the usual paradigm of keeping type checking separate from evaluation, i.e., after type checking succeeds, the evaluation is done on untyped terms. Accordingly, in Section \ref{subsec:evaluation-rules}, we provide evaluation rules for untyped terms (which could potentially lead to run-time errors). In Section \ref{subsec:type-preservation}, we prove that well-typed terms do not cause runtime errors.

\subsection{Evaluation rules}
\label{subsec:evaluation-rules}
\begin{figure}[t]
  \[\footnotesize
    \begin{tabular}{cccc}
      \\
      
      \infer[\textit{app}]{(C_1, \Sigma_1, M N)  \Downarrow (C_4, \Sigma_4, V')}
      {
      \begin{array}{c}
        (C_1, \Sigma_1, M) \Downarrow (C_2, \Sigma_2, \lambda x. M') \\
        (C_2, \Sigma_2, N) \Downarrow (C_3, \Sigma_3, V) \\
        (C_3, \Sigma_3, [V/x]M') \Downarrow (C_4, \Sigma_4, V')
      \end{array}
      }

      &
        \infer[\textit{force}]{(C, \Sigma, \force\ M)  \Downarrow (C'', \Sigma'', V)}
        {
        \begin{array}{c}
          (C, \Sigma, M) \Downarrow (C', \Sigma', \lift\ M')\\
          (C', \Sigma', M') \Downarrow (C'', \Sigma'', V)
        \end{array}
      }                   
      
      &
      \\
      \\
      \infer[\textit{let}]{(C, \Sigma, \mathsf{let}\ (x, y) = N\  \mathsf{in}\ M) \Downarrow (C'', \Sigma'', V)}
      {
      \begin{array}{c}
        (C, \Sigma, N) \Downarrow (C', \Sigma', (V_1, V_2))\\
        (C', \Sigma', [V_1/x, V_2/y]M) \Downarrow (C'', \Sigma'', V)
      \end{array}
      }
      &
        \infer[\textit{pair}]{(C, \Sigma, (M, N)) \Downarrow (C'', \Sigma'', (V_1, V_2))}
        {
        \begin{array}{c}
          (C, \Sigma, M) \Downarrow (C', \Sigma', V_1) \\
          (C', \Sigma', N) \Downarrow (C'', \Sigma'', V_2)
        \end{array}
      }        
      &
      \\ \\

      \infer[\textit{box}]
      {(C, \Sigma, \mathsf{box}_S\ M) \Downarrow (C', \Sigma', \circc(S, \interp{b}\circ D, U)) }
      {
      \begin{array}{c}
        (C, \Sigma, M) \Downarrow (C', \Sigma', \mathsf{lift} \ M') \\
        \gen(S) = (a, \Sigma'') \\
        (\interp{a}^{\dagger}, \Sigma'', M'\ a) \Downarrow (D, \Sigma''', b)\\
        \ungen(b, \Sigma''') = U
      \end{array}
      }
      &
        \infer[\textit{apply}]{(C_1, \Sigma_1, \mathsf{apply}(M, N))  \Downarrow (C', \Sigma', b)}
        {
        \begin{array}{c}
          (C_1, \Sigma_1, M) \Downarrow (C_2, \Sigma_2, \circc(S,D,U)) \\
          (C_2, \Sigma_2, N) \Downarrow (C_3, \Sigma_3, a) \\
          \gen(U) = (b, \Sigma)  \\          
          (C', \Sigma') = \mathrm{append}(C_3, \Sigma_3, a, D, b, \Sigma) 
        \end{array}
      }

      \\
      \\
      \infer[\textit{rev}]{(C, \Sigma, \reverse M) \Downarrow (C', \Sigma', \circc(U,D^{\dagger},S))}
      {(C, \Sigma, M) \Downarrow (C', \Sigma', \circc(S,D,U))}
      &
      \infer[\textit{ctrl}]{(C, \Sigma, \controlled_{K} M) \Downarrow (C', \Sigma',  \circc(\uu{K}\otimes S,\ctrled_{K}D,\uu{K}\otimes S))}
      {
        \begin{array}{c}
          (C, \Sigma, M) \Downarrow (C', \Sigma', \circc(S,D,U)) \\
        \end{array}
      }
    \end{tabular}
  \]
  \[\small
    \infer[\textit{wc}]{(C, \Sigma, \withCompute{M}{N}) \Downarrow
      (C'', \Sigma'', \circc(S,D_1 \action D_2,S))}
    {
      \begin{array}{c}
        (C', \Sigma', M) \Downarrow (C'', \Sigma'', \circc(S,D_1,U))\\
        (C, \Sigma, N) \Downarrow (C', \Sigma', \circc(U,D_2,U))
      \end{array}
    }          
  \]
  \caption{The operational semantics}
  \label{fig:operational}
\end{figure}

In this section, we will define a big-step, call-by-value operational
semantics for Proto-Quipper-C. Like the syntax and the type system,
the operational semantics is parameterized by the triple of categories
$\m$, $\R$, and $\N$. The operational semantics is defined in
Figure~\ref{fig:operational}.  It is defined on configurations that
are triples $(C, \Sigma, M)$, where $M$ is a term, $\Sigma$ is a label
context, and $C : X \to \interp{\Sigma}$ is a morphism in $\m$, $\R$
or $\N$. We call the pair $(C, \Sigma)$ the \textit{circuit state} of
the configuration. The evaluation of a closed program $M$
begins with the empty circuit state, i.e., $(\id_{I}, \emptyset)$. 

The \textit{app}, \textit{force}, \textit{let} and \textit{pair} rules
are standard. They do not directly modify the underlying circuit
state.
  
The \textit{box} and \textit{apply} rules use some notations that
warrant explaining. If $\Sigma$ is a label context, $a$ a simple term,
and $S$ a simple type, we write $\Sigma \vdash a : S$ instead of
$\Sigma \vdash_2 a : S$. In this case, $\Sigma$ and $S$ determine each
other uniquely given $a$. We write $\ungen(a,\Sigma)$ for the unique
$S$ such that $\Sigma\vdash a:S$. Conversely, we write $\gen(S) =
(a,\Sigma)$ to indicate the generation of a fresh simple term $a$ and
a corresponding label context $\Sigma$ such that $\Sigma \vdash a :
S$.  There is an obvious interpretation $\interp{a} :
\interp{\Sigma} \to \interp{S}$ as an isomorphism in $\m$, $\R$, or
$\N$.
  
In the \textit{box} rule, the codomain of $D$ is $\interp{\Sigma'''}$
and by the definition of $\ungen$, we have $\Sigma''' \vdash b : U$.
Then $\interp{S} \xstackrel{D}{\longrightarrow} \interp{\Sigma'''}
\xstackrel{\interp{b}}{\longrightarrow} \interp{U}$ is a morphism in
$\m$, $\R$ or $\N$.

In the \textit{rev} and \textit{ctrl} rules, the morphism $D$ is
reversed/controlled using the reverse/control function from the
category it belongs to. If $D$ is not a morphism in $\R$ or $\N$,
respectively, then it will cause a runtime error.  The type
preservation property (Theorem \ref{thm:type-preservation}) will
ensure that there are no such runtime errors. In the \textit{wc} rule,
we again abuse the notation; if $D_1\in\N$, we apply the functor
$G:\N\to\R$ before the ``$\action$'' operation. It is a runtime error
if $D_1\in\m$ or $D_2\not\in\N$.

In the \textit{apply} rule, by the definition of configurations, we
have $C_3:X\to\interp{\Sigma_3}$, for some object $X$. We also have
$D:\interp{S}\to\interp{U}$. Moreover, $\gen(U) = (b, \Sigma)$ implies
that $\Sigma \vdash b : U$ and $\interp{b} : \interp{\Sigma} \to
\interp{U}$. Let $\Sigma_3'$ be the unique label context such that
$\Sigma_3'\vdash a:S$. We assert that $\Sigma_3$ can be written in the
form $\Sigma_3 = \Sigma_3',\Sigma_3''$ (it is a runtime error if the
assertion fails). We have
$\interp{a}:\interp{\Sigma_3'}\to\interp{S}$. Note that
$(\interp{b}^{\dagger}\circ D\circ \interp{a}) : \interp{\Sigma_3'}
\to \interp{\Sigma}$.  Let $\Sigma'=\Sigma,\Sigma_3''$, and let
$C':X\to\interp{\Sigma}$ be the following composition:
\[
C' = X
\xrightarrow{C_3}\interp{\Sigma_3}
\xrightarrow{\iso}\interp{\Sigma_3'}\otimes\interp{\Sigma_3''}
\xrightarrow{(\interp{b}^{\dagger}\circ D\circ \interp{a})\otimes\id}\interp{\Sigma}\otimes\interp{\Sigma_3''}
\xrightarrow{\iso}\interp{\Sigma'}
~~ = ~~
\vcenter{
\Qcircuit @C=.4em @R=.3em {
    & \qw & \multigate{2}{C_{3}} &\qw &  \multigate{1}{\interp{b}^{\dagger}\circ D\circ \interp{a}}&\qw & \\
    & \qw & \ghost{C_{3}}&\qw    &  \ghost{\interp{b}^{\dagger}\circ D\circ \interp{a}} & \qw & \\
    & \qw & \ghost{C_{3}}&\qw    &  \qw & \qw &
}
}.
\]
We write $(C',\Sigma') = \mathrm{append}(C_3,\Sigma_3,a,D,b,\Sigma)$
for this operation. With a slight abuse of notation, this composition
is always possible, even when the morphisms belong to different
categories.  For example, if $D \in \R$ and $C_{3} \in \m$, we can
apply the functor $I:\R\to\m$ to $D$ before the composition.

\begin{example}
Since the evaluation rules of Figure~\ref{fig:operational} are a bit
complex, we give an example illustrating their use. Consider the
configuration
\[
(\id_{I},~ \emptyset,~ \boxt_{S} \lift (\lambda x . \mathbf{let}
(a_{1}, a_{2}) = x \ \mathbf{in} \ \apply(\circc(S,\cnot,S), (a_{2},
a_{1})))),
\]
where $S = \Qubit\otimes \Qubit$ and $\cnot : \Qubit\otimes \Qubit\to
\Qubit\otimes \Qubit$ is a morphism in $\N$.  Using the \textit{box}
rule, we first call $\gen(S)$, which returns a pair of fresh labels
$(\ell_{1},\ell_{2})$ and the label context $\Sigma=\ell_{1}:\Qubit,
\ell_{2}:\Qubit$. Next, using the rules \textit{app} and \textit{let},
we evaluate
\[
(\interp{(\ell_{1}, \ell_{2})}^{\dagger},~ \Sigma,~ (\lambda x . \mathbf{let}
(a_{1}, a_{2}) = x \ \mathbf{in} \ \apply(\circc(S,\cnot,S), (a_{2},
a_{1}))))(\ell_{1}, \ell_{2}))
\]
to
\[
((\interp{(\ell_{1}, \ell_{2})}^{\dagger},~\Sigma,~\apply(\circc(S,\cnot,S),
(\ell_{2}, \ell_{1}))).
\]
Then by the \textit{apply} rule, it is evaluated to
\[
(\interp{(\ell_{3},\ell_{4})}^{\dagger}\circ\cnot\circ\interp{(\ell_{2},
  \ell_{1})}\circ\interp{(\ell_{1}, \ell_{2})}^{\dagger},~ \Sigma',~ (\ell_{3},
\ell_{4})),
\]
where $\gen(S)$ returns fresh labels $(\ell_{3}, \ell_{4})$ and the
label context $\Sigma' = \ell_{3} : \Qubit, \ell_{4} : \Qubit$.
Therefore, continuing with the \textit{box} rule, we obtain
\[
(\id_{I},~\emptyset,~ \circc(S,
\interp{(\ell_{3}, \ell_{4})}\circ
\interp{(\ell_{3}, \ell_{4})}^{\dagger}\circ
\cnot\circ
\interp{(\ell_{2}, \ell_{1})}\circ
\interp{(\ell_{1}, \ell_{2})}^{\dagger}
,S).
\]
Note that $\interp{(\ell_{3}, \ell_{4})}\circ \interp{(\ell_{3},
  \ell_{4})}^{\dagger} = \id_{S}$ and $ \interp{(\ell_{2},
  \ell_{1})}\circ \interp{(\ell_{1}, \ell_{2})}^{\dagger} = \gamma :
\Qubit\otimes \Qubit\to\Qubit\otimes \Qubit$, where the morphism
$\gamma$ comes from the symmetric monoidal structure on $\N$. Thus,
the final configuration is
\[
  (\id_{I},~\emptyset,~ \circc(S,\cnot\circ\gamma,S)).    
\]
Note that the effect of the morphism $\gamma$ is to switch two wires
without inserting an explicit swap gate (it is, for example, the same
as the interpretation of the function $\tt f$ in
Section~\ref{ssec:cont-perm}). If we later control the circuit
$\circc(S,\cnot\circ\gamma,S)$, this morphism $\gamma$ will be
replaced by a controlled swap gate, just like circuit (b) in
Section~\ref{ssec:cont-perm}.
\end{example}

\subsection{Type preservation}
\label{subsec:type-preservation}
\noindent
In order to formulate type preservation, we first define a notion of
\textit{well-typed configuration}.

\begin{definition}[Well-typed configuration]
  We define a well-typed configuration $S \vdash_{\alpha} (C, \Sigma,
  M) : A; \Sigma'$ to mean: There exist $\beta$, $\gamma$, $\Sigma''$
  such that $C \in \D^{\beta}(\interp{S}, \interp{\Sigma})$,
  $\Sigma=(\Sigma'',\Sigma')$, $\Sigma'' \vdash_{\gamma} M : A$ and
  $\alpha=\beta\wedge\gamma$.
\end{definition}

\begin{theorem}[Type preservation]
  \label{thm:type-preservation}
  Suppose $S\vdash_{\alpha} (C_1, \Sigma_1, M) : A ; \Sigma'$ and
  $(C_1, \Sigma_1, M) \Downarrow (C_2, \Sigma_2, V)$. Then
  $S\vdash_{\alpha} (C_2, \Sigma_2, V) : A ; \Sigma'$.
\end{theorem}

\begin{proof}
  The proof is by induction on the derivation of
  $(C_1, \Sigma_1, M) \Downarrow (C_2, \Sigma_2, V)$.
\end{proof}

\section{A categorical semantics for reversing and control}
\label{sec:denotational-semantics}


\subsection{The abstract model}

As before, we assume that the categories $\m$, $\R$, and $\N$ are
given. These three categories lack the necessary structures to
interpret a programming language. Instead, we will interpret the
typing rules in a category $\A$, whose structure we now describe.

\begin{definition}
  \label{semantics-model}
  A category $\A$ is a \textit{model for Proto-Quipper-C} if it has
  the following structure.

  \begin{enumerate}
  \item \label{closed} $\A$ is symmetric monoidal closed. We write
    $\epsilon : (A\multimap B)\otimes A \to B$ for the application map.
    
  \item \label{coproducts} $\A$ has coproducts. Note that the tensor
    distributes over coproducts, because $-\otimes A$ is a left
    adjoint.
    
  \item \label{adj} There is an adjunction $p \dashv \flat : \set \to
    \A$ where $p$ is a strong monoidal functor.
    
  \item \label{monad-t} $\A$ is equipped with idempotent commutative
    strong monads $T_{0}$ and $T_{1}$ such that $T_{0} T_{1} \cong
    T_{0} \cong T_{1}T_{0}$.  More specifically, we require the
    natural maps $\eta_{T_{0}B} : T_{0}B \to T_{1}T_{0}B$ and
    $T_{0}\eta_{B} : T_{0}B \to T_{0}T_{1}B$ to be isomorphisms. By
    idempotent monad we mean that $\mu_{B} : T^{2}B\to TB$ is an
    isomorphism. We write $s : TA \otimes B \to T(A\otimes B)$ for the
    strength.
    
  \item \label{simple-types}

    There are full and faithful embeddings $\N\hookrightarrow\A$ ,
    $\R\hookrightarrow\Kl_{T_{1}}(\A)$, and
    $\m\hookrightarrow\Kl_{T_{0}}(\A)$.  These embedding functors are
    strong monoidal.  Moreover, the following diagram commutes for all
    $S, U \in \m$.
    \[ \footnotesize
      \begin{tikzcd}[cramped,row sep=small]
        \N(S, U)
        \arrow[r, "\cong"]
        \arrow[d, "~G"] &
        \A(S, U)
        \arrow[d, "~E"]\\
        \R(S, U)
        \arrow[r, "\cong"]
        \arrow[d, "~I"] &
        \Kl_{T_{1}}(\A)(S, U)
        \arrow[d, "~L"]\\
        \m(S, U)
        \arrow[r, "\cong"]& \Kl_{T_{0}}(\A)(S, U),
      \end{tikzcd}
    \]
    where $E$ and $L$ are the canonical identity-on-objects
    functors. Specifically, for any $f \in \A(A, B)$, we have $E(f) =
    \eta^{T_{1}} \circ f \in \Kl_{T_{1}}(\A)(A, B)$, and for any $f :
    A \to T_{1}B \in \Kl_{T_{1}}(\A)$, we have $L(f) = A
    \xstackrel{f}{\to} T_{1}B \xstackrel{\eta^{T_{0}}}{\to}
    T_{0}T_{1}B \xstackrel{(T_{0}\eta_{B})^{-1}}{\longrightarrow}
    T_{0}B$.
  \end{enumerate}      
\end{definition}

All models of Proto-Quipper support conditions
(\ref{closed})-(\ref{adj}).  Because $p : \set \to \A$ is a left
adjoint and is strong monoidal, we can deduce that
$
p(X) \cong p(\sum_{x\in X}1) \cong \sum_{x\in X}p(1) \cong  \sum_{x\in
  X}I.
$
Due to the adjunction $p \dashv \flat$, we also have
$
\set(X, \flat B)\cong \A(p(X), B) \cong \A(\sum_{x\in X}I, B) \cong \prod_{x\in X}\A(I, B)
\cong \set(X, \A(I, B)). 
$
Therefore by Yoneda's principle, we have $\flat(B) \cong \A(I, B)$.

For convenience, we define $T_2=\id$ to be the identity monad on $\A$.
Then in condition (\ref{monad-t}), the monads $T_0$, $T_1$, and $T_2$
represent (via their Kleisli categories) general circuits, reversible
circuits, and controllable circuits, respectively. Recall that we
write $\D^{\alpha}$ for the categories $\m$, $\R$, and $\N$ when
$\alpha=0$, $1$, and $2$, respectively. Condition (\ref{simple-types})
ensures that morphisms of simple types in the Kleisli category
$\Kl_{T_{\alpha}}(\A)$ correspond to morphisms in $\D^{\alpha}$.

For any $S, U\in \mathrm{obj}(\mathcal{\m})$, the maps $\unboxt$ and
$\boxt$ are defined by
\[\small
\begin{array}{@{}l}
  \unboxt : \D^{\beta}(S, U) \xstackrel{\cong}{} \A(S, T_{\beta}U)
  \xstackrel{\mathrm{curry}}{\to} \A(I, S \multimap T_{\beta}U)
  \xstackrel{\cong}{} \flat (S \multimap T_{\beta}U).
  \\
  \boxt : \flat T_{\alpha}(S \multimap T_{\beta}U)
  \xrightarrow{\iso}\set(1, \flat T_{\alpha}(S \multimap T_{\beta}U))
  \xrightarrow{\iso}\A(I, T_{\alpha}(S \multimap T_{\beta}U))
  \xrightarrow{k}\A(S, T_{\alpha\wedge\beta}U)
  \xrightarrow{\iso}\D^{\alpha \wedge\beta}(S, U).
\end{array}
\]
Here, $k:\A(I, T_{\alpha}(S \multimap T_{\beta}U))\to\A(S,
T_{\alpha\wedge\beta}U)$ is given by
\[
\A(I, T_{\alpha}(S \multimap T_{\beta}U))
\xrightarrow{}
\A(I, S \multimap T_{\alpha}T_{\beta}U)
\xrightarrow{}
\A(I, S \multimap T_{\alpha\wedge\beta}U)
\xrightarrow{}
\A(S, T_{\alpha\wedge\beta}U),
\]
where the first map comes from the strength.  Note that $\boxt$ is the
inverse of $\unboxt$ when $\alpha = 2$.

Given $\m, \R$ and $\N$, one can construct a
concrete category $\A$ satisfying the require properties,
using a generalization of our earlier work on biset enrichment \cite{FKRS-2023}. For
space reasons, we did not include the construction
here, but it can be found in the longer (and slightly outdated)
version of this paper at {\cite{FKRS2024-arxiv}}.


\subsection{Interpretation}

We will interpret the typing rules in the category $\A$.  The
following is the interpretation for types.
\[
\begin{array}{lll}
  \interp{\Qubit}& = & \Qubit\\
  \interp{A \otimes B}& = & \interp{A}\otimes \interp{B}\\
  \interp{\Circ_{\alpha}(S,U)}& = & p\D^{\alpha}(\interp{S}, \interp{U})\\
  
  \interp{!_{\alpha} A} & = & p\flat T_{\alpha}\interp{A}\\
  \interp{A\multimap_{\alpha} B} &  = & \interp{A} \multimap T_{\alpha} \interp{B}
\end{array}
\]
We interpret a typing context $\Gamma$ as the tensor of all of its
objects (denoted by $\interp{\Gamma}$).  Each valid typing judgment
$\Gamma \vdash_{\alpha} M : A$ will be interpreted as a morphism
$\interp{M} : \interp{\Gamma} \to T_{\alpha} \interp{A}$ in $\A$.
The interpretation is defined by induction on the typing rules.
See Appendix~\ref{app:interpretation} for the details.

\subsection{Soundness of operational semantics}

We now prove that the operational semantics is sound with respect to
the categorical model $\A$. To do so, we first interpret a well-typed
configuration as a morphism in $\A$.

\begin{definition}
  Suppose $S \vdash_{\alpha} (C, \Sigma, M) : A ; \Sigma'$ is a
  well-typed configuration, with $C \in \D^{\beta}(\interp{S},
  \interp{\Sigma})$, $\Sigma=(\Sigma'',\Sigma')$, $\Sigma''
  \vdash_{\gamma} M : A$ and $\alpha=\beta\wedge\gamma$.
  We define $\interp{(C, \Sigma, M)} : \interp{S} \to
  T_{\alpha}(\interp{A}\otimes \interp{\Sigma'})$ as follows:
  \[
  \interp{S}
  \xrightarrow{C}
  T_{\beta}(\interp{\Sigma})
  \xrightarrow{\iso}
  T_{\beta}(\interp{\Sigma''}\otimes \interp{\Sigma'})
  \xrightarrow{T_{\beta}( \interp{M}\otimes \interp{\Sigma'})}
  T_{\beta}(T_{\gamma}\interp{A} \otimes \interp{\Sigma'})
  \xrightarrow{}
  T_{\alpha}(\interp{A}\otimes \interp{\Sigma'}),\]
  where the last step uses the strength and the natural isomorphism
  $T_{\beta}T_{\gamma}\iso T_{\alpha}$.
\end{definition}

\begin{theorem}[Soundness of the evaluation]
  \label{thm:soundness}
  If $S \vdash_{\alpha} (C, \Sigma, M) : A ; \Sigma'$ and $(C, \Sigma,
  M) \Downarrow (C', \Sigma', V)$, then $\interp{(C, \Sigma, M)} =
  \interp{(C', \Sigma', V)}$.
\end{theorem}

\begin{proof}
  See Appendix~\ref{app:proof-soundness}.
\end{proof}

We remark that our semantics also satisfies computational adequacy,
but in our setting, it is a trivial consequence of soundness. This is
because if $V$ is a value of type $\Circ(S,U)$, then $V$ is by
definition a circuit, i.e., a morphism of the appropriate category
$\m$, $\R$, or $\N$. Its semantics is essentially itself. It
immediately follows that $\interp{V}=\interp{V'}$ implies
$V=V'$. Since our language is also terminating, adequacy for any
closed term of type $\Circ(S,U)$ then follows from soundness.

\section{Conclusion}
\label{conclude}

In this paper, we showed how to extend Proto-Quipper with reversing,
control, and the with-computed operation. Our language is
parameterized by three categories $\m$, $\R$, and $\N$, which
correspond to general quantum circuits, reversible circuits, and
controllable circuits, respectively.  We defined a type system that
uses modalities to distinguish the different types of circuits. We
provided an operational and a denotational semantics for the language;
the latter takes the form of an abstract categorical model in which
our modalities are represented by monads. We proved that the
operational semantics is sound with respect to the categorical model.

There are many possible directions for future work.  For example, in
this paper, we only considered the modalities of reversibility and
controllability.  But it seems that our construction of the type
system and its categorical semantics would easily generalize from a
three element set to an arbitrary poset of modalities. Introducing
additional modalities might be useful for characterizing additional
properties of gates.  Although we have made a prototype implementation
that includes a type inference algorithm, we did not formally study
its properties.  Another challenge for future work is to combine
modalities with dependent types.  Although our software implementation
does support both modalities and dependent types, we have not
considered a formal semantics for combining them.

\section*{Acknowledgements}

We thank the referees for their thoughtful comments. This work was
supported by the Natural Sciences and Engineering Research Council of
Canada (NSERC) and by the Air Force Office of Scientific Research
under Award No.\@ FA9550-21-1-0041.

\bibliographystyle{eptcs}
\bibliography{modalities}
\newpage
\appendix

\section{Controlling a CCZ gate}
\label{app:controlling-ccz}

Here, we give a basic example of using control and with-computed in
the prototype implementation of Proto-Quipper from {\cite{prototype}}.
It is well-known that a CCZ gate can be implemented by 7 T-gates with
T-depth one \cite{selinger2013quantum}. See the following circuit.

\[
\includegraphics[scale=1]{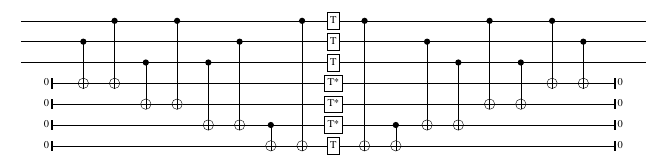}
\]
\noindent We can define the above circuit in Proto-Quipper, using its
\texttt{withComputed} operator.

{\small
\begin{verbatim}
my_ccz : Circ(Qubit * Qubit * Qubit, Qubit * Qubit * Qubit)
my_ccz = withComputed box_cnot_circuit box_parallel_T
\end{verbatim}
}

The function \texttt{box\_parallel\_T} generates the 7 parallel T
gates in the middle of the circuit and the function
\texttt{box\_cnot\_circuit} generates the initialization gates and the
cnot circuits before the parallel T gates. Their definitions are
available in \texttt{examples/CCZ.dpq} in {\cite{prototype}}.

We can use Proto-Quipper's \texttt{control} operator to add an an
extra control to the CCZ circuit.

{\small
\begin{verbatim}
ctrl_ccz : Circ(Qubit * (Qubit * Qubit * Qubit), Qubit * (Qubit * Qubit * Qubit))
ctrl_ccz = control my_ccz
\end{verbatim}
}

\noindent This generates the following circuit. We can see that only
the T-gates in the middle are controlled.
\[
  \includegraphics[scale=0.8]{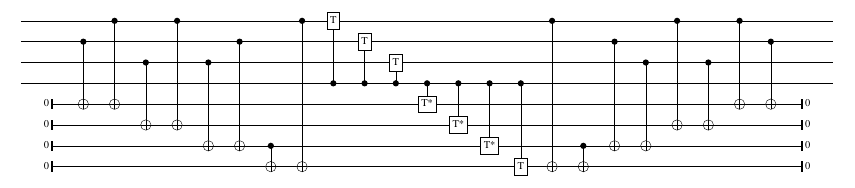}
\]

If instead, we program the CCZ circuit without using the
\texttt{withComputed} operation, we will get an error when trying to
control it. In the program below, \texttt{my\_ccz'} is defined by
manual composition, where \texttt{cnot\_circuit\_rev} is the reverse
version of \texttt{cnot\_circuit}.

{\small
\begin{verbatim}
cnot_circuit_rev : 
  !(Qubit * Qubit * Qubit * Qubit * Qubit * Qubit * Qubit 
    -> Qubit * Qubit * Qubit)
cnot_circuit_rev = unbox (reverse (boxCirc cnot_circuit))

my_ccz' : Circ(Qubit * Qubit * Qubit, Qubit * Qubit * Qubit)
my_ccz' = 
  boxCirc $ \input -> cnot_circuit_rev (parallel_T (cnot_circuit input))

-- The following gives rise to a typing error.
ctrl_my_ccz' : Circ(Qubit * (Qubit * Qubit * Qubit), Qubit * (Qubit * Qubit * Qubit))
ctrl_my_ccz' = control my_ccz'
\end{verbatim}
}

Controlling the circuit \texttt{my\_ccz'} gives rise to a typing
error, because according to the semantics of \texttt{control}, we
would have to control all the gates in the circuit \texttt{my\_ccz'},
which includes the non-controllable initialization gates. So in order
to control the CCZ circuit that has initialization gates, we must
construct the circuit using the \texttt{withComputed} operation.

\section{The categorical interpretation of the typing rules}
\label{app:interpretation}

Here, we only give some important cases. The remaining typing rules can be interpreted similarly.
  \begin{itemize}
  \item Case
  \[
    \infer{\Gamma \vdash_\alpha \mathsf{box}_{S} M :
      \Circ_{\alpha_1 \wedge \alpha_2}(S,U)}
    {\Gamma \vdash_\alpha M : {!}_{\alpha_1} (S \multimap_{\alpha_2} U)}
  \]
  We define $\interp{\mathsf{box}_{S} M}$ to be
  \[\interp{\Gamma}\xstackrel{\interp{M}}{\to} T_{\alpha} p\flat
T_{\alpha_{1}}(\interp{S} \multimap T_{\alpha_{2}}\interp{U})
\xstackrel{T_{\alpha}p \boxt}{\to} T_{\alpha}
p\D^{\alpha_{1}\wedge\alpha_{2}}(\interp{S}, \interp{U}).\] 
    
  \item Case
    \[
      \infer[]{\Phi \vdash_{2} \circc(S, C, U) :  \Circ_{\alpha}(S, U)}
      {
      \begin{array}{ll}
        C \in \D^{\alpha}(\interp{S}, \interp{U})
      \end{array}
      }
    \]
    
    Since $C : \interp{S} \to \interp{U}$ is
    a morphism in $\D^{\alpha}$, 
     we define $\interp{\Phi \vdash_{2} \circc(S, C, U) :  \Circ_{\alpha}(S, U)}$
     as the following composition, where $1_{C} : 1 \to \D^{\alpha}(\interp{S}, \interp{U})$
     such that $1_{C}(*) = C$ and $\mathrm{discard} = p !_{X}, !_{X}\in \set(X, 1)$.
    
    \[ \interp{\Phi} \xstackrel{\discard}{\longrightarrow} p1 \xstackrel{p 1_{C}}{\longrightarrow} p\D^{\alpha}(\interp{S}, \interp{U}).\]

  \item Case
    \[ \infer{\Gamma \vdash_\alpha \controlled_{K} M : \Circ_{2}(\uu{K}\otimes U, \uu{K}\otimes U)}{\Gamma \vdash_\alpha M : \Circ_{2}(U, U)}
\]

By the induction hypothesis, we have $\interp{M} :\interp{\Gamma} \to
T_{\alpha} p\N(\interp{S}, \interp{S}) $. 
Using the function $\ctrled_{K} : \N(\interp{U}, \interp{U})\to \N(\interp{\uu{K}}\otimes \interp{U}, \interp{\uu{K}}\otimes \interp{U})$,
we define $\interp{\controlled_{K} M}$ to be the following.
    \[\interp{\Gamma} \xstackrel{\interp{M}}{\longrightarrow} 
      T_{\alpha}p\N(\interp{U}, \interp{U})
      \xstackrel{T_{\alpha}p(\ctrled_{K})}{\longrightarrow}
      T_{\alpha}p\N( \interp{\uu{K}} \otimes \interp{U}, \interp{\uu{K}}\otimes \interp{U})\]

  \item Case
    \[
      \infer{\Gamma_{1}, \Gamma_{2} \vdash_{\alpha \wedge \beta} \withComputed M N : \Circ_2(U, U)}
      {
        \begin{array}{ll}
          \Gamma_{1} \vdash_\alpha M : \Circ_{\gamma}(U, S) & \gamma > 0 \\
          \Gamma_{2} \vdash_\beta N : \Circ_2(S, S)
        \end{array}
         }  
    \]

    Suppose $\alpha = \beta = \gamma = 2$
    and $\Gamma_{1} \vdash_2 M : \Circ_{2}(U, S)$.
    By the induction hypothesis, we have
    \[pG_{U,S} \circ \interp{M} : \interp{\Gamma_{1}} \to p\N(\interp{U}, \interp{S}) \to p\R(\interp{U}, \interp{S}), \]
    \[\interp{N} : \interp{\Gamma_{2}} \to p\N(\interp{S}, \interp{S}). \]

    Using the function $\withComputed :  \R(\interp{U}, \interp{S}) \times \N(\interp{S}, \interp{S}) \to \N(\interp{U}, \interp{U})$, we define

    $\interp{\withComputed M N}$ to be the following.
    \[\interp{\Gamma_{1}}\otimes \interp{\Gamma_{2}}
      \xstackrel{(pG_{U,S} \circ \interp{M})\otimes \interp{N}}{\longrightarrow}
      p\R(\interp{U}, \interp{S}) \otimes p\N(\interp{S}, \interp{S}) 
    \]
    \[
         \xstackrel{p(\withComputed)}{\longrightarrow}         p\N(\interp{U}, \interp{U})
    \]
    Suppose $\alpha = \beta = 2$, $\gamma = 1$ and  $\Gamma_{1} \vdash_2 M : \Circ_{1}(U, S)$.
    By the induction hypothesis, we have
    \[\interp{M} : \interp{\Gamma_{1}} \to p\R(\interp{U}, \interp{S}), \]
    \[\interp{N} : \interp{\Gamma_{2}} \to p\N(\interp{S}, \interp{S}). \]
    We define
    $\interp{\withComputed M N}$ to be the following.
    \[\interp{\Gamma_{1}}\otimes \interp{\Gamma_{2}}
      \xstackrel{\interp{M}\otimes \interp{N} }{\longrightarrow}
      p\R(\interp{U}, \interp{S}) \otimes p\N(\interp{S}, \interp{S}) 
   \xstackrel{p(\withComputed)}{\longrightarrow}
      p\N(\interp{U}, \interp{U})
      \]
      The remaining cases (e.g. when $\alpha,\beta\neq 2$) are similar.
  \end{itemize}

\section{Proof of soundness}
\label{app:proof-soundness}

\begin{theorem}[Soundness of the evaluation]
  If $S \vdash_{\alpha} (C, \Sigma, M) : A ; \Sigma'$ and $(C, \Sigma,
  M) \Downarrow (C', \Sigma', V)$, then $\interp{(C, \Sigma, M)} =
  \interp{(C', \Sigma', V)}$.
\end{theorem}

The proof is by induction on $(C,\Sigma,M) \Downarrow (C',\Sigma',
V)$. We only give the proof for the rules \textit{box}, \textit{apply},
\textit{ctrl}, and \textit{wc}, as the remaining rules are routine.
\begin{itemize}
  \item Case
    \[
    \infer[\textit{box}]{(C, \Sigma, \mathsf{box}_{S} M) \Downarrow (C', \Sigma', \circc(S, \interp{b}\circ D, U)) }
          {
            \begin{array}{c}
              (C, \Sigma, M) \Downarrow (C', \Sigma', \mathsf{lift} \ M') \\
              \gen(S) = (a, \Sigma'') \\
              (\interp{a}^{\dagger}, \Sigma'', M'\ a) \Downarrow (D, \Sigma''', b)\\
              \ungen(b, \Sigma''') = U
            \end{array}
          }      
          \]
          Suppose $S' \vdash_{2} (C, \Sigma, \boxt_{S} M) : \Circ_{0}(S, U); \Sigma'_{2}$.
          This implies that $C \in \N(\interp{S}, \interp{\Sigma'_{1}}\otimes \interp{\Sigma'_{2}})$ and $\Sigma_{1}' \vdash_{2} \boxt_{S} M : \Circ_{0}(S, U)$.
          Thus $\Sigma_{1}' \vdash_{2} M : {!}_{0}(S \multimap_{2} U)$, $C'\in  \N(\interp{S}, \interp{\Sigma_{2}'})$,
          and $\vdash_{0} M' : S \multimap_{2} U$. 
          By the induction hypothesis, $\interp{(C, \Sigma, M)} = \interp{(C', \Sigma', \lift M')}$ and $\interp{(\interp{a}^{\dagger}, \Sigma'', M'a)} = \interp{(D,\Sigma''', b)}$. Thus we have
          \[
          \interp{S'} \xstackrel{C}{\to} \interp{\Sigma'_{1}}\otimes\interp{\Sigma'_{2}}
          \xstackrel{\interp{M}\otimes \interp{\Sigma'_{2}}}{\to} p\flat T_{0}(\interp{S}\multimap \interp{U}) \otimes \interp{\Sigma'_{2}} 
          \]
          \[
          = \interp{S'} \xstackrel{C'}{\to} I\otimes\interp{\Sigma'_{2}}\xstackrel{p\delta\interp{M'}\otimes \interp{\Sigma'_{2}}}{\to} p\flat T_{0}(\interp{S}\multimap \interp{U}) \otimes \interp{\Sigma'_{2}} 
          \]
          and

          \begin{align*}
            I \otimes \interp{S} \xstackrel{I \otimes \interp{a}^{\dagger}}{\to} & I\otimes \interp{\Sigma''} \xstackrel{\interp{M'}\otimes \interp{a}}{\to} T_{0}(\interp{S} \multimap \interp{U})\otimes \interp{S}\\
            & \xstackrel{s}{\to} T_{0}(\interp{S} \multimap \interp{U}\otimes \interp{S}) \xstackrel{T_{0}\epsilon}{\to} T_{0}\interp{U}
          \end{align*}
          $=$
          \begin{align*}
            I\otimes \interp{S} \xstackrel{\interp{M'}\otimes \interp{S}}{\to} & T_{0}(\interp{S} \multimap \interp{U})\otimes \interp{S} \\
            & \xstackrel{s}{\to} T_{0}(\interp{S} \multimap \interp{U}\otimes \interp{S}) \xstackrel{T_{0}\epsilon}{\to} T_{0}\interp{U}
          \end{align*}
          $\stackrel{(*)}{=}$
          \begin{align*}
            \interp{S} \xstackrel{D}{\to} T_{0}\interp{\Sigma'''} \xstackrel{T_{0}\interp{b}}{\to} T_{0}\interp{U}.
          \end{align*}
          Moreover,
          $\interp{(C, \Sigma, \boxt_{S} M)}=$
          \begin{align*}
            \interp{S'} \xstackrel{C}{\to} & 
            \interp{\Sigma_{1}'} \otimes \interp{\Sigma_{2}'}\xstackrel{\interp{M}\otimes \interp{\Sigma_{2}'}}{\to} p\flat T_{0}(\interp{S}\multimap \interp{U})\otimes \interp{\Sigma_{2}'}\\
            & \xstackrel{p \boxt}{\to}  p\m(\interp{S}, \interp{U})\otimes \interp{\Sigma_{2}'} 
          \end{align*}
          $=$
          \begin{align*}
            \interp{S'} \xstackrel{C'}{\to}& I\otimes\interp{\Sigma'_{2}}\xstackrel{p\delta\interp{M'}\otimes \interp{\Sigma'_{2}}}{\to} p\flat T_{0}(\interp{S}\multimap \interp{U}) \otimes \interp{\Sigma'_{2}} \\
            & \xstackrel{p \boxt}{\to} p\m(\interp{S}, \interp{U})\otimes \interp{\Sigma_{2}'}.
          \end{align*}
          So we just need to show
          \[I\xstackrel{p\delta\interp{M'}}{\to} p\flat T_{0}(\interp{S}\multimap \interp{U}) 
          \xstackrel{p \boxt}{\to} p\m(\interp{S}, \interp{U}) = I \xstackrel{p1_{\interp{b}\circ D}}{\to} p\m(\interp{S}, \interp{U}).\]
          In other words, we need to show $\interp{b}\circ D  = \boxt(\delta \interp{M'})$.
          It suffices to show they are equal in $\A$. 
          This is true by definition of $\boxt$ and $(*)$. 

        \item Case
          \[
          \infer[\textit{apply}]{(C_1, \Sigma_1, \mathsf{apply}(M, N))  \Downarrow (C', \Sigma_{5}, b)}
                {
                  \begin{array}{c}
                    (C_1, \Sigma_1, M) \Downarrow (C_2, \Sigma_2, \circc(S, D, U)) \\
                    (C_2, \Sigma_2, N) \Downarrow (C_3, \Sigma_{3}, a) \\
                    \gen(U) = (b, \Sigma_{4})  \\          
                    (C', \Sigma_{5}) = \mathrm{append}(C_3,\Sigma_{3}, a, D, b, \Sigma_{4}) 
                  \end{array}
                }
                \]

                Suppose $S' \vdash_{0} (C_{1}, \Sigma_{1}, \apply(M, N)) : U; \Sigma'$, where
                $\Sigma_{1} = (\Sigma'', \Sigma')$, $C_{1} \in \R(\interp{S}, \interp{\Sigma''}\otimes \interp{\Sigma'})$, $\Sigma'' \vdash_{0} \apply(M, N) : U$,
                $\Sigma_{1}'' \vdash_{0} M : \Circ_{1}(S, U)$ and $\Sigma_{2}'' \vdash_{1} N : S$ and $\Sigma'' = (\Sigma_{1}'', \Sigma_{2}''), \Sigma_{2} = (\Sigma_{2}'', \Sigma'), \Sigma_{3} = (\Sigma_{2}''', \Sigma')$. By the induction hypothesis, we know that
                $\interp{(C_{1},\Sigma_{1}, M)} = \interp{(C_{2}, \Sigma_{2}, \circc(S, D, U))}$ and $\interp{(C_{2}, \Sigma_{2}, N)} = \interp{(C_{3}, \Sigma_{3}, a)}$.
                Thus
                \begin{align*}
                  \interp{S'} & \xstackrel{C_{1}}{\longrightarrow} T_{1}(\interp{\Sigma_{1}''}\otimes\interp{\Sigma_{2}''}\otimes \interp{\Sigma'})\\
                  &  \xstackrel{T_{1}(\interp{M}\otimes\interp{\Sigma_{2}''}\otimes \interp{\Sigma'})}{\longrightarrow}
                  T_{1}(T_{0}p\N(\interp{S}, \interp{U}) \otimes \interp{\Sigma_{2}''}\otimes \interp{\Sigma'}) \\
                  & \xstackrel{T_{1}s}{\longrightarrow}T_{1}T_{0}(p\N(\interp{S}, \interp{U}) \otimes \interp{\Sigma_{2}''}\otimes \interp{\Sigma'}) \\
                  & \xstackrel{\cong}{\longrightarrow} T_{0}(p\N(\interp{S}, \interp{U}) \otimes \interp{\Sigma_{2}''}\otimes \interp{\Sigma'})
                \end{align*}
                $=$
                \begin{align*}
                  \interp{S'} \xstackrel{C_{2}}{\longrightarrow}T_{0}(\interp{\Sigma_{2}''}\otimes \interp{\Sigma'})
                  \xstackrel{T_{0}(p1_{D}\otimes \interp{\Sigma_{2}''}\otimes \interp{\Sigma'})}{\longrightarrow} T_{0}(p\N(\interp{S}, \interp{U})\otimes \interp{\Sigma_{2}''}\otimes \interp{\Sigma'})
                \end{align*}
                and
                \begin{align*}
                  \interp{S'} &\xstackrel{C_{2}}{\longrightarrow}T_{0}(\interp{\Sigma_{2}''}\otimes \interp{\Sigma'})
                  \xstackrel{T_{0}(\interp{N}\otimes \interp{\Sigma'})}{\longrightarrow} T_{0}(T_{1}\interp{S}\otimes \interp{\Sigma'})\\
                  & \xstackrel{T_{0}s}{\longrightarrow}T_{0}T_{1}(\interp{S}\otimes \interp{\Sigma'})
                  \xstackrel{\cong}{\longrightarrow} T_{0}(\interp{S}\otimes \interp{\Sigma'})
                \end{align*}
                $=$
                \begin{align*}
                  \interp{S'} \xstackrel{C_{3}}{\longrightarrow}T_{0}(\interp{\Sigma_{2}'''}\otimes \interp{\Sigma'})
                  \xstackrel{T_{0}(\interp{a}\otimes \interp{\Sigma'})}{\longrightarrow} T_{0}(\interp{S}\otimes \interp{\Sigma'}). 
                \end{align*}
                We want to show that $\interp{(C_{1},\Sigma_{1}, \apply(M, N))} = \interp{(C', \Sigma_{5}, b)}$, where

                $(C', \Sigma_{5}) = \mathrm{append}(C_3, \Sigma_{3}, a, D, b, \Sigma_{4}) $ is the following morphism in $\m$. So $\Sigma_{5} = (\Sigma_{4}, \Sigma')$.
                \begin{align*}
                  \interp{S'} & \xstackrel{C_{3}}{\longrightarrow}
                  \interp{\Sigma_{2}'''}\otimes \interp{\Sigma'}
                  \xstackrel{\interp{a}\otimes \interp{\Sigma'}}{\longrightarrow} \interp{S} \otimes \interp{\Sigma'}\\
                  & 
                  \xstackrel{G(D)\otimes \interp{\Sigma'}}{\longrightarrow} \interp{U} \otimes \interp{\Sigma'}
                  \xstackrel{\interp{b}^{\dagger}\otimes \interp{\Sigma'}}{\longrightarrow} \interp{\Sigma_{4}} \otimes \interp{\Sigma'}.
                \end{align*}
                Since $\m \hookrightarrow \Kl_{T_{0}}(\A)$ and $\N \hookrightarrow \A$, the
                corresponding morphism in $\A$ is 
                \begin{align*}
                  \interp{S'} & \xstackrel{C_{3}}{\longrightarrow}
                  T_{0}(\interp{\Sigma_{2}'''}\otimes \interp{\Sigma'})
                  \xstackrel{T_{0}(\interp{a}\otimes \interp{\Sigma'})}{\longrightarrow} T_{0}(\interp{S} \otimes \interp{\Sigma'}) \\
                  & \xstackrel{T_{0}(D\otimes \interp{\Sigma'})}{\longrightarrow} T_{0}(\interp{U} \otimes \interp{\Sigma'})
                  \xstackrel{T_{0}(\interp{b}^{\dagger}\otimes \interp{\Sigma'})}{\longrightarrow} T_{0}(\interp{\Sigma_4} \otimes \interp{\Sigma'}).
                \end{align*}
                We have
                $RHS =$

                \begin{align*}
                  \interp{S'} \xstackrel{C_{3}}{\longrightarrow}
                  T_{0}(\interp{\Sigma_{2}'''}\otimes \interp{\Sigma'})
                  \xstackrel{T_{0}(\interp{a}\otimes \interp{\Sigma'})}{\longrightarrow} T_{0}(\interp{S} \otimes \interp{\Sigma'})
                  \xstackrel{T_{0}(D\otimes \interp{\Sigma'})}{\longrightarrow} T_{0}(\interp{U} \otimes \interp{\Sigma'})
                \end{align*}
                and $LHS =$
                \begin{align*}
                  \interp{S'} & \xstackrel{C_{1}}{\longrightarrow}
                  T_{1}(\interp{\Sigma_{1}''}\otimes\interp{\Sigma_{2}''}\otimes \interp{\Sigma'})\\
                  & \xstackrel{T_{1}(\interp{M}\otimes \interp{N}\otimes \interp{\Sigma'})}{\longrightarrow}
                  T_{1}(T_{0}p\N(\interp{S}, \interp{U})\otimes T_{1}\interp{S}\otimes \interp{\Sigma'})\\
                  & \xstackrel{T_{1}s}{\longrightarrow} T_{1}T_{0}(p\N(\interp{S}, \interp{U})\otimes T_{1}\interp{S}\otimes \interp{\Sigma'}) \xstackrel{T_{0}s}{\longrightarrow}T_{0}T_{1}(p\N(\interp{S}, \interp{U})\otimes \interp{S}\otimes \interp{\Sigma'}) \\
                  & \xstackrel{T_{0}(\unboxt\otimes \interp{S}\otimes\interp{\Sigma'})}{\longrightarrow} T_{0}(p\flat(\interp{S}\multimap \interp{U})\otimes \interp{S}\otimes \interp{\Sigma'}) \\
                  & \xstackrel{T_{0}(\force\otimes \interp{S}\otimes\interp{\Sigma'})}{\longrightarrow} T_{0}((\interp{S}\multimap \interp{U})\otimes \interp{S}\otimes \interp{\Sigma'}) \xstackrel{T_{0}(\epsilon\otimes\interp{\Sigma'})}{\longrightarrow} T_{0}( \interp{U}\otimes \interp{\Sigma'}) 
                \end{align*}
                $=$
                
                \begin{align*}
                  \interp{S'} & \xstackrel{C_{2}}{\longrightarrow}T_{0}(\interp{\Sigma_{2}''}\otimes \interp{\Sigma'}) \\
                  & \xstackrel{T_{0}(p1_{D}\otimes \interp{N}\otimes \interp{\Sigma'})}{\longrightarrow} T_{0}(p\N(\interp{S}, \interp{U})\otimes T_{1}\interp{S}\otimes \interp{\Sigma'})\\
                  & \xstackrel{T_{0}s}{\longrightarrow}T_{0}T_{1}(p\N(\interp{S}, \interp{U})\otimes \interp{S}\otimes \interp{\Sigma'}) \\
                  & \xstackrel{T_{0}(\unboxt \otimes \interp{S}\otimes\interp{\Sigma'})}{\longrightarrow} T_{0}(p\flat(\interp{S}\multimap \interp{U})\otimes \interp{S}\otimes \interp{\Sigma'}) \\
                  & \xstackrel{T_{0}(\force\otimes \interp{S}\otimes\interp{\Sigma'})}{\longrightarrow} T_{0}((\interp{S}\multimap \interp{U})\otimes \interp{S}\otimes \interp{\Sigma'}) \xstackrel{T_{0}(\epsilon\otimes\interp{\Sigma'})}{\longrightarrow} T_{0}( \interp{U}\otimes \interp{\Sigma'}) 
                \end{align*}
                $=$
                \begin{align*}
                  \interp{S'} & \xstackrel{C_{3}}{\longrightarrow}T_{0}(\interp{\Sigma_{2}'''}\otimes \interp{\Sigma'})
                  \xstackrel{T_{0}(p1_{D}\otimes \interp{a}\otimes \interp{\Sigma'})}{\longrightarrow} T_{0}(p\N(\interp{S}, \interp{U})\otimes\interp{S}\otimes \interp{\Sigma'})
                  \\
                  & \xstackrel{T_{0}(\unboxt\otimes \interp{S}\otimes\interp{\Sigma'})}{\longrightarrow} T_{0}(p\flat(\interp{S}\multimap \interp{U})\otimes \interp{S}\otimes \interp{\Sigma'}) 
                  \\
                  & \xstackrel{T_{0}(\force\otimes \interp{S}\otimes\interp{\Sigma'})}{\longrightarrow} T_{0}((\interp{S}\multimap \interp{U})\otimes \interp{S}\otimes \interp{\Sigma'}) \xstackrel{T_{0}(\epsilon\otimes\interp{\Sigma'})}{\longrightarrow} T_{0}( \interp{U}\otimes \interp{\Sigma'}). 
                \end{align*}

                So we just need to show
                \begin{align*}
                  \interp{S} \xstackrel{D}{\to} \interp{U} 
                \end{align*}
                $=$
                \begin{align*}
                  I \otimes \interp{S} & \xstackrel{1_{D}}{\to} p\N(\interp{S}, \interp{U})\otimes \interp{S}
                  \xstackrel{p \unboxt\otimes \interp{S}}{\to} p\flat(\interp{S}\multimap \interp{U}) \otimes \interp{S}\\
                  &  \xstackrel{\force \otimes \interp{S}}{\to} (\interp{S}\multimap \interp{U}) \otimes \interp{S} \xstackrel{\epsilon}{\to} \interp{U}.
                \end{align*}
                This is true because $LHS =$
                \begin{align*}
                  I\otimes \interp{S} \xstackrel{\mathrm{curry}(D)\otimes \interp{S}}{\to} (\interp{S}\multimap \interp{U})\otimes \interp{S} \xstackrel{\epsilon}{\to} \interp{U}
                \end{align*}
                and $RHS =$
                \begin{align*}
                  I\otimes \interp{S} &
                  \xstackrel{p (\delta \curry{D})  \otimes \interp{S}}{\to} p\flat (\interp{S}\multimap \interp{U})\otimes \interp{S}\\
                  & \xstackrel{\force \otimes \interp{S}}{\to}(\interp{S}\multimap \interp{U})\otimes \interp{S} \xstackrel{\epsilon}{\to} \interp{U}
                \end{align*}
                and we have $\force \circ p \delta = \id$. 

\item Case
  \[
  \infer[ctrl]{(C, \Sigma, \controlled_{K} M) \Downarrow (C', \Sigma''',  \circc(\uu{K}\otimes U, \ctrled_{K}D, \uu{K}\otimes U))}
        {
          \begin{array}{c}
            (C, \Sigma, M) \Downarrow (C', \Sigma''', \circc(U, D, U)) \\
          \end{array}
        }
        \]
        \begin{itemize}
        \item Suppose $C : \interp{S} \to \interp{\Sigma''}\otimes \interp{\Sigma'} \in \N$,
          $\Sigma = (\Sigma'', \Sigma')$ and
          $\Sigma'' \vdash_{2} M : \Circ_{2}(U, U)$. 
          By the induction hypothesis and type preservation (Theorem \ref{thm:type-preservation}), we have

          \[\interp{(C, \Sigma, M)} = \interp{(C' , \Sigma', \circc(U, D, U))}\]
          and $S \vdash_{2} (C', \Sigma', \circc(U, D, U)) : \Circ_{2}(U, U) : \Sigma'$, where $\vdash_{2} \circc(U, D, U) : \Circ_{2}(U, U)$ and $\Sigma'''= \Sigma'$.
          So 
          \[
          \interp{S} \xstackrel{C}{\to}\interp{\Sigma''} \otimes \interp{\Sigma'}\xstackrel{\interp{M}\otimes \interp{\Sigma'}}{\to} p\N(\interp{U}, \interp{U})\otimes \interp{\Sigma'}
          \]
          \[
          = \interp{S} \xstackrel{C'}{\to}I \otimes \interp{\Sigma'}\xstackrel{p1_{D} \otimes \interp{\Sigma'}}{\to} p\N(\interp{U}, \interp{U})\otimes \interp{\Sigma'}.
          \]
          The above equality implies the following. 
          \[
          \interp{S} \xstackrel{C}{\to}\interp{\Sigma''} \otimes \interp{\Sigma'}\xstackrel{\interp{M}\otimes \interp{\Sigma'}}{\to} p\N(\interp{U}, \interp{U})\otimes \interp{\Sigma'}
          \]
          \[
          \quad\quad\quad\quad \xstackrel{{p(\ctrled_{K})\otimes \interp{\Sigma'}}}{\to}
          p\N(\interp{\uu{K}}\otimes \interp{U}, \interp{\uu{K}}\otimes \interp{U})\otimes \interp{\Sigma'}
          \]
          \[
          = \interp{S} \xstackrel{C'}{\to}I \otimes \interp{\Sigma'}\xstackrel{p1_{D}\otimes \interp{\Sigma'}}{\to} p\N(\interp{U}, \interp{U})\otimes \interp{\Sigma'}
          \]
          \[\quad\quad\quad\quad \xstackrel{p(\ctrled_{K})\otimes \interp{\Sigma'}}{\to}
          p\N(\interp{\uu{K}}\otimes \interp{U}, \interp{\uu{K}}\otimes \interp{U})\otimes \interp{\Sigma'}.
          \]

          Thus $\interp{(C, \Sigma, \controlled_{K} M)} =  \interp{(C', \Sigma', \circc(\uu{K}\otimes U, \ctrled_{K}(D), \uu{K}\otimes U))}$. 
          
        \item Now suppose $\Sigma = (\Sigma'', \Sigma')$, $C \in \m(\interp{S}, \interp{\Sigma''}\otimes \interp{\Sigma'})$ and
          $\Sigma'' \vdash_{1} M : \Circ_{2}(U, U)$.
          By the induction hypothesis, we have $\interp{(C, \Sigma, M)} = \interp{(C' , \Sigma''', \circc(U, D, U))}$, i.e.,
          \[
          \interp{S} \xstackrel{C}{\to}T_{0}(\interp{\Sigma''} \otimes \interp{\Sigma'})\xstackrel{T_{0}(\interp{M}\otimes \interp{\Sigma'})}{\to} T_{0}(T_{1}p\N(\interp{U}, \interp{U})\otimes \interp{\Sigma'})
          \]
          \[
          \xstackrel{T_{0}s}{\to} T_{0}(p\N(\interp{U}, \interp{U})\otimes \interp{\Sigma'})
          \]
          \[
          = \interp{S} \xstackrel{C'}{\to} T_{0}(I \otimes \interp{\Sigma'})\xstackrel{T_{0}(p1_{D}\otimes \interp{\Sigma'})}{\to} T_{0}(p\N(\interp{U}, \interp{U})\otimes \interp{\Sigma'}).
          \]
          Note that $\Sigma'''= \Sigma'$. The above equality implies the following. 
          \[
          \interp{S} \xstackrel{C}{\to}T_{0}(\interp{\Sigma''} \otimes \interp{\Sigma'})\xstackrel{T_{0}(\interp{M}\otimes \interp{\Sigma'})}{\to} T_{0}(T_{1}p\N(\interp{U}, \interp{U})\otimes \interp{\Sigma'})
          \]
          \[
          \xstackrel{T_{0}s}{\to} T_{0}(p\N(\interp{U}, \interp{U})\otimes \interp{\Sigma'})
          \]
          \[
          \xstackrel{T_{0}(p(\ctrled_{K}))\otimes \interp{\Sigma'})}{\to} T_{0}(p\N(\interp{\uu{K}}\otimes \interp{U}, \interp{\uu{K}} \otimes \interp{U})\otimes \interp{\Sigma'})
          \]
          \[
          = \interp{S} \xstackrel{C'}{\to} T_{0}(I \otimes \interp{\Sigma'})\xstackrel{T_{0}(p1_{D}\otimes \interp{\Sigma'})}{\to} T_{0}(p\N(\interp{U}, \interp{U})\otimes \interp{\Sigma'})
          \]
          \[
          \xstackrel{T_{0}(p(\ctrled_{K})\otimes \interp{\Sigma'})}{\to} T_{0}(p\N(\interp{\uu{K}}\otimes \interp{U}, \interp{\uu{K}} \otimes \interp{U})\otimes \interp{\Sigma'}).
          \]
          By the naturality of the strength $s$, we have \[\interp{(C, \Sigma, \controlled_{K} M)} =  \interp{(C', \Sigma''', \circc(\uu{K}\otimes U, \ctrled_{K}D, \uu{K}\otimes U))}.\]
        \end{itemize}
        
      \item Case
        \[
        \infer[wc]{(C, \Sigma, \withCompute{M}{N}) \Downarrow
          (C'', \Sigma_{2}, \circc(U, D_1 \action D_2, U ))}
              {
                \begin{array}{c}
                  (C, \Sigma, M) \Downarrow (C', \Sigma_{1}, \circc(U, D_1, S'))\\
                  (C', \Sigma_{1}, N) \Downarrow (C'', \Sigma_{2}, \circc(S', D_2, S'))
                \end{array}
              }          
              \]
              \begin{itemize}
              \item Suppose 
                $\Sigma = (\Sigma_{1}'', \Sigma_{2}'', \Sigma')$, $C \in \N(\interp{S}, \interp{\Sigma_{1}''}\otimes \interp{\Sigma_{2}''}\otimes \interp{\Sigma'})$, $\Sigma_{1}'' \vdash_{2} N : \Circ_{2}(S', S')$ and $\Sigma_{2}'' \vdash_{2} M : \Circ_{1}(U, S')$. By the induction hypothesis, we have
                \[\interp{(C, \Sigma, M)} = \interp{(C', \Sigma_{1}, \circc(U, D_{1}, S'))}\] and
                $\interp{(C', \Sigma_{1}, N)} = \interp{(C'', \Sigma_{2}, \circc(S', D_{2}, S'))}$.
                Thus we have 
                \[
                \interp{S} \xstackrel{C}{\to} \interp{\Sigma''_{1}}\otimes \interp{\Sigma_{2}''} \otimes \interp{\Sigma'} \xstackrel{\interp{M}\otimes \interp{\Sigma_{2}''} \otimes \interp{\Sigma'}}{\to}
                p\R(\interp{U}, \interp{S'}) \otimes \interp{\Sigma_{2}''} \otimes \interp{\Sigma'}
                \]
                \[
                = \interp{S} \xstackrel{C'}{\to} I\otimes \interp{\Sigma_{2}''} \otimes \interp{\Sigma'} \xstackrel{p1_{D_{1}}\otimes \interp{\Sigma_{2}''} \otimes \interp{\Sigma'}}{\to}
                p\R(\interp{U}, \interp{S'}) \otimes \interp{\Sigma_{2}''} \otimes \interp{\Sigma'}
                \]
                and
                \[
                \interp{S} \xstackrel{C'}{\to} \interp{\Sigma_{2}''} \otimes \interp{\Sigma'} \xstackrel{\interp{N} \otimes \interp{\Sigma'}}{\to}
                p\N(\interp{S'}, \interp{S'}) \otimes \interp{\Sigma'}
                \]
                \[
                = \interp{S} \xstackrel{C''}{\to} I\otimes \interp{\Sigma'} \xstackrel{p1_{D_{2}} \otimes \interp{\Sigma'}}{\to}
                p\N(\interp{S'}, \interp{S'}) \otimes \interp{\Sigma'}.
                \]
                We want to show $\interp{(C, \Sigma, \withCompute{M}{N})} = \interp{(C'', \Sigma_{2}, \circc(U, D_{1}\action D_{2}, U))}$, i.e.,
                \[
                \interp{S} \xstackrel{C}{\to}\interp{\Sigma''_{1}}\otimes \interp{\Sigma_{2}''} \otimes \interp{\Sigma'} \xstackrel{\interp{M}\otimes \interp{N} \otimes \interp{\Sigma'}}{\longrightarrow}
                \]
                \[
                p\R(\interp{U}, \interp{S}) \otimes p\N(\interp{S}, \interp{S})  \otimes \interp{\Sigma'}\xstackrel{p(\withComputed) \otimes \interp{\Sigma'}}{\longrightarrow} p\N(\interp{U}, \interp{U}) \otimes \interp{\Sigma'}
                \]

                \[
                = \interp{S} \xstackrel{C''}{\to} I\otimes \interp{\Sigma'} \xstackrel{p1_{D_{1} \action D_{2}} \otimes \interp{\Sigma'}}{\to} p\N(\interp{U}, \interp{U}) \otimes \interp{\Sigma'}.
                \]

                This is the case by the induction hypotheses and $\withComputed \circ (1_{D_{1}}\otimes 1_{D_{2}}) = 1_{D_{1} \action D_{2}}$. 
              \end{itemize}
              
  \end{itemize}

\end{document}